\documentclass[journal,draftcls,onecolumn,12pt]{IEEEtran}

\usepackage{amsmath}
\usepackage{amssymb}
\usepackage[dvips]{graphicx}
\usepackage{curves,color}
\usepackage{subfigure}
\usepackage{dsfont,slashbox,wasysym}
\usepackage{tikz}
\usepackage{pdfpages,enumerate}
\usepackage{multirow}

\allowdisplaybreaks

\newtheorem{theorem}{Theorem}

\newtheorem{definition}{Definition}

\begin{document}
\title{Construction of High-Rate Regular Quasi-Cyclic LDPC Codes Based on Cyclic Difference Families}
\author{Hosung~Park, Seokbeom~Hong, Jong-Seon~No,~\IEEEmembership{Fellow,~IEEE,} and~Dong-Joon~Shin,~\IEEEmembership{Senior~Member,~IEEE}%
\thanks{H.~Park, S.~Hong, and J.-S.~No are with the Department of Electrical Engineering and Computer Science, INMC, Seoul National University, Seoul 151-744, Korea (e-mail: hpark1@snu.ac.kr, fousbyus@ccl.snu.ac.kr, jsno@snu.ac.kr).}%
\thanks{D.-J. Shin is with the Department of Electronic Engineering, Hanyang University, Seoul 133-791, Korea (e-mail: djshin@hanyang.ac.kr).}%
}%


\maketitle

\begin{abstract}
For a high-rate case, it is difficult to randomly construct good low-density parity-check (LDPC) codes of short and moderate lengths because their Tanner graphs are prone to making short cycles. Also, the existing high-rate quasi-cyclic (QC) LDPC codes can be constructed only for very restricted code parameters. In this paper, a new construction method of high-rate regular QC LDPC codes with parity-check matrices consisting of a single row of circulants with the column-weight 3 or 4 is proposed based on special classes of cyclic difference families. The proposed QC LDPC codes can be constructed for various code rates and lengths including the minimum achievable length for a given design rate, which cannot be achieved by the existing high-rate QC LDPC codes. It is observed that the parity-check matrices of the proposed QC LDPC codes have full rank.
It is shown that the error correcting performance of the proposed QC LDPC codes of short and moderate lengths is almost the same as that of the existing ones through numerical analysis.
\end{abstract}


\vspace{2mm}
\section{Introduction}\label{sec:introduction}
\IEEEPARstart{L}{ow}-density parity-check (LDPC) codes \cite{Gallager} have been one of dominant error-correcting codes for high-speed communication systems or data storage systems because they asymptotically have capacity-approaching performance under iterative decoding with moderate complexity. Among them, quasi-cyclic (QC) LDPC codes are well suited for hardware implementation using simple shift registers due to the regularity in their parity-check matrices so that they have been adopted in many practical applications.

For a high-rate case, it is difficult to randomly construct good LDPC codes of short and moderate lengths because their parity-check matrices are so dense compared to a low-rate case for the given code length and degree distribution that they are prone to making short cycles.
Among well-known structured LDPC codes, finite geometry LDPC codes \cite{Kou}-\cite{Kamiya} and LDPC codes constructed from combinatorial designs \cite{Vasic}-\cite{Fujisawa} are adequate for high-rate LDPC codes. The error correcting performance of these LDPC codes is verified under proper decoding algorithms but they have severe restrictions on flexibly choosing the code rate and length. Also, since finite geometry LDPC codes usually have much redundancy and large weights in their parity-check matrices, they are not suitable for a strictly power-constrained system with iterative message-passing decoding.

It is known that the parity-check matrix structure consisting of a single row of circulants \cite{Tang}-\cite{Johnson2}, \cite{Xia}, \cite{Baldi1} is adequate for generating high-rate QC LDPC codes of short and moderate lengths. The class-I circulant EG-LDPC codes in \cite{Tang} and the duals of one-generator QC codes in \cite{Kamiya} are constructed from the Euclidean geometry and the affine geometry, respectively, and they have very restricted rates and lengths and much  redundancy. In \cite{Vasic}-\cite{Johnson2}, QC LDPC codes constructed from cyclic difference families (CDFs) \cite{Colbourn1} are proposed, which also have restricted lengths. Computer search algorithms are proposed in \cite{Xia} and \cite{Baldi1} for various-length QC LDPC codes of this parity-check matrix structure, but they cannot generate QC LDPC codes as short as the QC LDPC codes constructed from CDFs for most of code rates.

In this paper, new high-rate regular QC LDPC codes with parity-check matrices consisting of a single row of circulants with the column-weight 3 or 4 are proposed based on special classes of CDFs. In designing the proposed QC LDPC codes, we can flexibly choose the code rate and length including the minimum achievable code length for a given design rate. The parity-check matrices of the proposed QC LDPC codes have full rank when the column-weight is 3 and have almost full rank when the column-weight is 4 because there is just one redundant row. Numerical analysis shows that the error correcting performance of the proposed QC LDPC codes of short and moderate lengths is almost the same as that of the existing high-rate QC LDPC codes.

The remainder of the paper is organized as follows.
Section \ref{sec:CDF} introduces the definition and the existence theorems of CDFs, and provides a construction method of a class of CDFs. In Section \ref{sec:proposed}, high-rate regular QC LDPC codes are proposed and analyzed. In Section \ref{sec:simulation}, the error correcting performance of the proposed QC LDPC codes is compared to that of the existing high-rate QC LDPC codes via numerical analysis. Finally, the conclusions are provided in Section \ref{sec:conclusions}.

\vspace{2mm}
\section{Cyclic Difference Families} \label{sec:CDF}

\subsection{Definition and Existence}

A cyclic difference family is defined as follows.

\vspace{2mm}
\begin{definition}[\cite{Colbourn1}]
Consider the additive group $\mathbb{Z}_v = \{ 0,1,\ldots,v-1 \}$.
Then $t$ $k$-element subsets of $\mathbb{Z}_v$, $B_i = \{ b_{i1}, b_{i2},\ldots,b_{ik} \}$, $i=1,2,\ldots,t$, $b_{i1} < b_{i2} < \cdots < b_{ik}$, form a $(v,k,\lambda)$ \textit{cyclic difference family} (CDF) if every nonzero element of $\mathbb{Z}_v$ occurs $\lambda$ times among the differences $b_{im}-b_{in}$, $i=1,2,\ldots,t$, $m \neq n$, $m,n=1,2,\ldots,k$.
\end{definition}

\vspace{2mm}
According to \cite{Vasic}-\cite{Fujisawa}, $(k(k-1)t+1,k,1)$ CDFs are adequate for constructing parity-check matrices of QC LDPC codes with girth at least 6. Theorem \ref{thm:CDF_existence} shows the existence of such CDFs.

\vspace{2mm}
\begin{theorem}
The existence of $(k(k-1)t+1,k,1)$ CDFs is given as:
\begin{itemize}
	\item[1)] There exists a $(6t+1,3,1)$ CDF for all $t \geq 1$ \cite{Peltesohn}.
	\item[2)] A $(12t+1,4,1)$ CDF exists for all $1 \leq t \leq 1000$ \cite{Ge}. 
	\item[3)] A $(20t+1,5,1)$ CDF exists for $1  \leq t \leq 50$ and $t \neq 16,25,31,34,40,45$ \cite{Julian}.
\end{itemize}
\label{thm:CDF_existence}
\end{theorem}

\vspace{2mm}
In this paper, a special class of CDFs, called perfect difference family (PDF), will be used to construct high-rate regular QC LDPC codes with parity-check matrices consisting of a single row of circulants so that more various code parameters can be achieved. Before introducing its definition, we will define two terms as follows.
Consider $t$ $k$-element subsets of $\mathbb{Z}_v$, $B_i = \{ b_{i1}, b_{i2},\ldots,b_{ik} \}$, $i=1,2,\ldots,t$, $b_{i1} < b_{i2} < \cdots < b_{ik}$. Among the differences $b_{im}-b_{in}$, $i=1,2,\ldots,t$, $m \neq n$, $m,n=1,2,\ldots,k$, we will call $tk(k-1)/2$ differences $b_{im} - b_{in}$, $i=1,2,\ldots,t$, $1 \leq m < n \leq k$, as the \textit{forward differences over $\mathbb{Z}_v$} of the subsets and the remaining $tk(k-1)/2$ differences as the \textit{backward differences over $\mathbb{Z}_v$} of the subsets.

\vspace{2mm}
\begin{definition}[\cite{Colbourn1}]
Consider a $(v,k,1)$ CDF, $B_i = \{ b_{i1}, b_{i2},\ldots,b_{ik} \}$, $i=1,2,\ldots,t$, $b_{i1} < b_{i2} < \cdots < b_{ik}$, with $v=k(k-1)t+1$. Then the CDF is called a $(v,k,1)$ \textit{perfect difference family} (PDF) if the $tk(k-1)/2$ backward differences cover the set $\{ 1,2,\ldots,(v-1)/2 \}$.
\end{definition}

\vspace{2mm}
The condition on the existence of PDFs is stricter than that of CDFs. Some recent results on the existence of PDFs are summarized in the following theorem.

\vspace{2mm}
\begin{theorem}
The existence of $(k(k-1)t+1,k,1)$ PDFs is given as:
\begin{itemize}
	\item[1)] A $(6t+1,3,1)$ PDF exists if and only if $t \equiv 0 ~\mathrm{or}~ 1 \mod 4$ \cite{Beth}.
	\item[2)] A $(12t+1,4,1)$ PDF exists for $t=1$, $4 \leq t \leq 1000$ \cite{Ge}.
	\item[3)] $(20t+1,5,1)$ PDFs are known for $t=6,8,10$ but for no other small value of $t$ \cite{Mathon}.
	\item[4)] There is no $(k(k-1)t+1,k,1)$ PDF for $k \geq 6$ \cite{Mathon}.
\end{itemize}
\label{thm:PDF_existence}
\end{theorem}

\vspace{2mm}
Since there are no PDFs for $k \geq 6$ and no sufficiently many PDFs for $k = 5$ from Theorem \ref{thm:PDF_existence}, we focus on the case of $k=3,4$. The construction of PDFs for $k=3$ and $4$ is provided in \cite{Colbourn1} and \cite{Ge}, respectively.

\subsection{Construction of $(6t+1,3,1)$ CDFs for $t \equiv 2~\mathrm{or}~3 \mod 4$}
\label{subsec:Skolem}

Although $(6t+1,3,1)$ PDFs  do not exist for $t \equiv 2~\mathrm{or}~3 \mod 4$, a class of $(6t+1,3,1)$ CDFs constructed from hooked Skolem sequences for $t \equiv 2~\mathrm{or}~3 \mod 4$ can be used to construct parity-check matrices of QC LDPC codes with various code parameters, which consist of a single row of circulants.

\vspace{2mm}
\begin{definition}[\cite{Colbourn1}]
A \textit{Skolem sequence} of order $t$ is a sequence $S = ( a_1,a_2,\ldots,a_{2t} )$ of integers satisfying the following two conditions:
\begin{itemize}
	\item[i)] For every $k \in \{ 1,2,\ldots,t \}$, there exist exactly two elements $a_i$ and $a_j$ in $S$ such that $a_i = a_j = k$.
	\item[ii)] If $a_i = a_j = k$ with $i < j$, then $j-i = k$.
\end{itemize}
Skolem sequences are also written as collections of ordered pairs $\{ (u_i,v_i): 1 \leq i \leq t,v_i - u_i = i \}$ with $\bigcup_{i=1}^{t} \{ u_i,v_i \} = \{ 1,2,\ldots,2t \}$. A \textit{hooked Skolem sequence} of order $t$ is a sequence $S = ( a_1,a_2,\ldots,a_{2t-1},a_{2t}=0,a_{2t+1} )$ satisfying the conditions i) and ii).
\label{def:Skolem}
\end{definition}

\vspace{2mm}
A hooked Skolem sequence of order $t$ exists if and only if $t \equiv 2~ \mathrm{or} ~3 \mod 4$, and it can be constructed by the method in \cite{Colbourn1} as follows. Note that the ordered pairs are used to represent a hooked Skolem sequence as in Definition \ref{def:Skolem}.

\vspace{2mm}
\begin{itemize}
	\item $t=2$; $(1,2),(3,5)$ 
	\item $t=3$; $(1,4),(2,3),(5,7)$ 
	\item $t=4s+2$, $s \geq 1$;
	\item[] $
		\begin{cases}
			(r,4s-r+2),&r=1,\ldots,2s\\
			(4s+r+3,8s-r+4),&r=1,\ldots,s-1\\
			(5s+r+2,7s-r+3),&r=1,\ldots,s-1\\
			(2s+1,6s+2),(4s+2,6s+3),\\
			(4s+3,8s+5),(7s+3,7s+4)
		\end{cases}
		$
	\item $t=4s-1$, $s \geq 2$;
	\item[] $
		\begin{cases}
			(4s+r,8s-r-2),&r=1,\ldots,2s-2\\
			(r,4s-r-1),&r=1,\ldots,s-2\\
			(s+r+1,3s-r),&r=1,\ldots,s-2\\
			(s-1,3s),(s,s+1),(2s,4s-1),\\
			(2s+1,6s-1),(4s,8s-1)
		\end{cases}
		$
\end{itemize}

\vspace{2mm}
From hooked Skolem sequences, $(6t+1,3,1)$ CDFs can be constructed for $t \equiv 2~ \mathrm{or} ~3 \mod 4$. After constructing a hooked Skolem sequence $(u_i,v_i)$ of order $t$, $1 \leq i \leq t$, a $(6t+1,3,1)$ CDF is obtained by letting $B_i = \{ 0, i, v_i +t\}$. We can easily check that all differences in $B_i$'s cover the set $\{ 1,2,\ldots,6t \}$.

\vspace{2mm}
\section{High-Rate QC LDPC Codes Constructed From PDFs and CDFs} \label{sec:proposed}

\subsection{Proposed QC LDPC Codes}

Consider a binary regular LDPC code whose parity-check matrix $H$ is a $1 \times L$ array of $z \times z$ circulants given as
\begin{equation}
	H = \begin{bmatrix} H_1 & H_2 & \cdots & H_{L} \end{bmatrix}
\label{eq:H}
\end{equation}
where a \textit{circulant} $H_{i}$ is defined as a matrix whose each row is a cyclic shift of the row above it. This LDPC code is quasi-cyclic because applying circular shifts within each length-$z$ subblock of a codeword gives another codeword.
A circulant is entirely described by the positions of nonzero elements in the first column. Let $i$, $0 \leq i \leq z-1$, be the index of the $(i+1)$-st element in the first column. Then, the \textit{shift value}(s) of a circulant is (are) defined as the index (indices) of the nonzero element(s) in the first column. Note that a shift value takes the value from $0$ to $z-1$. Let $d_v$ ($d_c$) denote the column-weight (row-weight) of $H$. Then we have $d_c = d_v L$ and the design rate of this LDPC code is $R=(L-1)/L$. Let $s_{ij}$, $i=1,\ldots,L$ and $j=1,\ldots,d_v$, denote the $j$-th smallest shift value of $H_i$, that is, $s_{i1} < s_{i2} < \cdots < s_{id_v}$, which correspond to the indices of 1's in the first column of $H_i$.

We propose a new class of high-rate QC LDPC codes which have the parity-check matrix form in (\ref{eq:H}) constructed from PDFs or CDFs given in Subsection \ref{subsec:Skolem}. Under some proper constraints, a parity-check matrix of QC LDPC code with the column-weight 3 or 4 can be obtained by taking shift values of $H_i$ from $B_i$ in the CDF or the PDF. More concretely, QC LDPC codes can be constructed by using $B_i$, $i=1,\ldots,L$, of CDF or PDF as follows:

\vspace{2mm}
\begin{itemize}
	\item[1)] Choose the code parameters $d_v=3~\mathrm{or}~4$, $L$, and $z$ such that
	\begin{itemize}
		\item[i)] $L \geq 2$ for $d_v =3$ and $4 \leq L \leq 1000$ for $d_v=4$
		\item[ii)] $z \geq d_v(d_v-1)L+1$, where $z \neq 6L+2$ for $d_v=3$ and $L \equiv 2~ \mathrm{or} ~3 \mod 4$.
	\end{itemize}
	\item[2)] If $d_v=3$ and $L \equiv 2~ \mathrm{or} ~3 \mod 4$, construct a $(6L+1,3,1)$ CDF $B_i$, $i=1,\ldots,L$, as in Subsection \ref{subsec:Skolem}. Otherwise, construct a $(d_v (d_v-1) L+1,d_v,1)$ PDF $B_i$, $i=1,\ldots,L$, as in \cite{Colbourn1} and \cite{Ge}.
	\item[3)] Let $s_{ij} = b_{ij}$, $i=1,2,\ldots,L$ and $j=1,2,\ldots,d_v$.

\end{itemize}

\vspace{2mm}
For $d_v=5$, QC LDPC codes with parity-check matrices consisting of a single row of circulants can also be constructed by the proposed procedure, but for $L$ other than 6, 8, and 10, $(20L+1,5,1)$ PDFs are still unknown as in Theorem \ref{thm:PDF_existence}.

\subsection{Girth of the Proposed QC LDPC Codes}

It is well known that parity-check matrices including a circulant of column-weight larger than or equal to 3 have a girth at most 6 \cite{Tanner}. In this subsection, we will show that the proposed QC LDPC codes have the girth 6 by proving that there is no cycle of length 4 in the parity-check matrices. Let $\delta_{ijk}^{(z)}$ denote the difference $s_{ij} - s_{ik} \mod z$ of shift values $s_{ij}$ and $s_{ik}$ in $H_i$.

\vspace{2mm}
\begin{theorem}
Consider a $(d_v(d_v-1)L+1,d_v,1)$ PDF $B_i = \{ b_{i1}, b_{i2},\ldots, b_{id_v} \}$, $i=1,2,\ldots,L$. The proposed QC LDPC codes constructed from this PDF do not have any cycle of length 4 for every $z \geq d_v(d_v-1)L+1$.
\label{thm:main1}
\end{theorem}
\begin{proof}
The necessary and sufficient condition for avoiding cycles of length 4 is that all $\delta_{ijk}^{(z)}$'s, $i=1,\ldots,L$ and $1 \leq j \neq k \leq d_v$, are distinct. The backward differences over $\mathbb{Z}_z$ of the shift values cover $1$ to $d_v (d_v-1)L/2$ due to the property of the PDF, and thus the forward differences over $\mathbb{Z}_z$ of the shift values cover $z-d_v (d_v-1)L/2$ to $z-1$. Since $z \geq d_v(d_v-1)L+1$, all $\delta_{ijk}^{(z)}$'s are distinct.
\end{proof}

\vspace{2mm}
\begin{theorem}
For $L \equiv 2~ \mathrm{or} ~3 \mod 4$, consider a $(6L+1,3,1)$ CDF $B_i = \{ b_{i1}, b_{i2}, b_{i3} \}$, $i=1,2,\ldots,L$, constructed from a hooked Skolem sequence of order $L$. The proposed QC LDPC codes constructed from this CDF do not have any cycle of length 4 for every $z \geq 6L+1$ and $z \neq 6L+2$.
\label{thm:main2}
\end{theorem}
\begin{proof}
We only need to show that all $\delta_{ijk}^{(z)}$'s, $i=1,\ldots,L$ and $1 \leq j \neq k \leq d_v$, are distinct for $z \geq 6L+3$. We can see from Definition \ref{def:Skolem} that the maximum value of $v_i$ in a hooked Skolem sequence of order $L$ is $2L +1$. Since $B_i = \{ 0, i, v_i+L \}$, the minimum and the maximum of the backward differences over $\mathbb{Z}_z$ of the shift values are $1$ and $3L+1$, respectively. Since all backward differences over $\mathbb{Z}_z$ are distinct and take values $1$ through $3L+1$, every forward difference over $\mathbb{Z}_z$ has a value from $z-3L-1$ to $z-1$. Therefore, all $\delta_{ijk}^{(z)}$'s are distinct for $z \geq 6L+3$.
\end{proof}

\vspace{2mm}
Note that for $d_v=3$, $L \equiv 2~ \mathrm{or} ~3 \mod 4$, and $z = 6L+2$, a forward difference and a backward difference over $\mathbb{Z}_z$ of the shift values in the proposed QC LDPC codes have $3L+1$ in common. Thus it seems that another construction method is needed for $z = 6L+2$. However, in fact, there does not exist any construction method which avoids cycles of length 4 as shown in the following theorem.  

\vspace{2mm}
\begin{theorem}
For $d_v=3$, $L \equiv 2~ \mathrm{or} ~3 \mod 4$, and $z = 6L+2$, which is not the case covered in Theorem \ref{thm:main2}, QC LDPC codes whose parity-check matrices have the form in (\ref{eq:H}) cannot avoid cycles of length 4 for any shift value assignment.
\label{thm:main3}
\end{theorem}
\begin{proof}
Assume that there exists a shift value assignment such that the parity-check matrix avoids cycles of length 4. If $\delta_{ijk}^{(z)} = 3L+1$ for some $i$, $j$, $k$, the difference $\delta_{ikj}^{(z)}$ is also $3L+1$. Therefore, all $6L$ differences $\delta_{ijk}^{(z)}$, $i=1,2,\ldots,L$ and $1 \leq j \neq k \leq 3$, have to cover $1$ to $6L+1$ except $3L+1$. 

Let $\Delta_{B}$ denote the sum of backward differences over $\mathbb{Z}_z$ of the shift values. Note that the addition is calculated over the integer. Then, $\Delta_{B}$ is odd because
\begin{align}
\Delta_{B} &= \sum_{i,j,k:~j > k,\atop \delta_{ijk}^{(z)}<3L+1} {\delta_{ijk}^{(z)}} + \sum_{i,j,k:~j > k,\atop \delta_{ijk}^{(z)}>3L+1} {\delta_{ijk}^{(z)}} \nonumber \\
&\equiv \left\{ \sum_{i,j,k:~j > k,\atop \delta_{ijk}^{(z)}<3L+1} {\delta_{ijk}^{(z)}} - \sum_{i,j,k:~j > k,\atop \delta_{ijk}^{(z)}>3L+1} {\delta_{ijk}^{(z)}} \right\}~~\mathrm{mod}~2 \nonumber \\
&\equiv \left\{ \sum_{i,j,k:~j > k,\atop \delta_{ijk}^{(z)}<3L+1} {\delta_{ijk}^{(z)}} + \sum_{i,j,k:~j > k,\atop \delta_{ijk}^{(z)}>3L+1} {(6L+2 - \delta_{ijk}^{(z)})} \right\}~~\mathrm{mod}~2 \nonumber \\
&\equiv \left\{ \sum_{i,j,k:~j > k,\atop \delta_{ijk}^{(z)}<3L+1} {\delta_{ijk}^{(z)}} + \sum_{i,j,k:~j < k,\atop \delta_{ijk}^{(z)}<3L+1} {\delta_{ijk}^{(z)}} \right\}~~\mathrm{mod}~2 \nonumber \\
&\equiv \sum_{i=1}^{3L} {i}~~\mathrm{mod}~2 \nonumber \\
&\equiv 1~~\mathrm{mod}~2.
\label{eq:diff_sum}
\end{align}
On the other hand, since $\Delta_{B}$ can be expressed as
\begin{align*}
\Delta_{B} &= \sum_{i} \sum_{j,k:~j>k} \delta_{ijk}^{(z)} \\
&= \sum_{i} \sum_{j,k:~j>k} {(s_{ij}-s_{ik})} \\
&= \sum_{i} 2(s_{i3} - s_{i1}),
\end{align*}
this contradicts (\ref{eq:diff_sum}) in that the parities of $\Delta_{B}$ are different. Therefore, there is no shift value assignment such that the parity-check matrix in (\ref{eq:H}) avoids cycles of length 4 for $d_v=3$, $L \equiv 2~ \mathrm{or} ~3 \mod 4$, and $z = 6L+2$.
\end{proof}

\subsection{Advantages of the Proposed QC LDPC Codes}

The proposed QC LDPC codes have advantages mainly on being able to have various code parameters while guaranteeing the girth 6. For $d_v=3$, $L$ only has to be larger than or equal to 2 and for $d_v=4$, $L$ can be any integer from 4 to 1000, which are the same as the conventional QC LDPC codes constructed from CDFs \cite{Vasic}-\cite{Johnson2} except for $d_v=4$ and $L=2,3$. Moreover, for a fixed $L$, $z$ can have any value as long as $z \geq d_v(d_v-1)L+1$ except for the case of $d_v=3$, $L \equiv 2~ \mathrm{or} ~3 \mod 4$, and $z = 6L+2$.
On the other hand, the conventional QC LDPC codes constructed from CDFs generally do not guarantee the girth 6 for $z > d_v(d_v-1)L+1$ for a given $L$. In fact, it can be easily seen that $z$ can have the arbitrary value from $z \geq d_v (d_v-1)L +1$ and $z \equiv 1 \mod d_v (d_v -1)$ for the conventional QC LDPC codes to have the girth 6 because for the given CDF $B_i$, $i=1,\ldots,L$, the conventional QC LDPC codes can be constructed by only using the sets $B_{i_1},B_{i_2},\ldots,B_{i_{L'}}$ for $L' < L$. However, the conventional QC LDPC codes cannot still have sufficiently various code lengths.


\begin{figure*}[tb]
	\centering
	\includegraphics[scale=0.6]{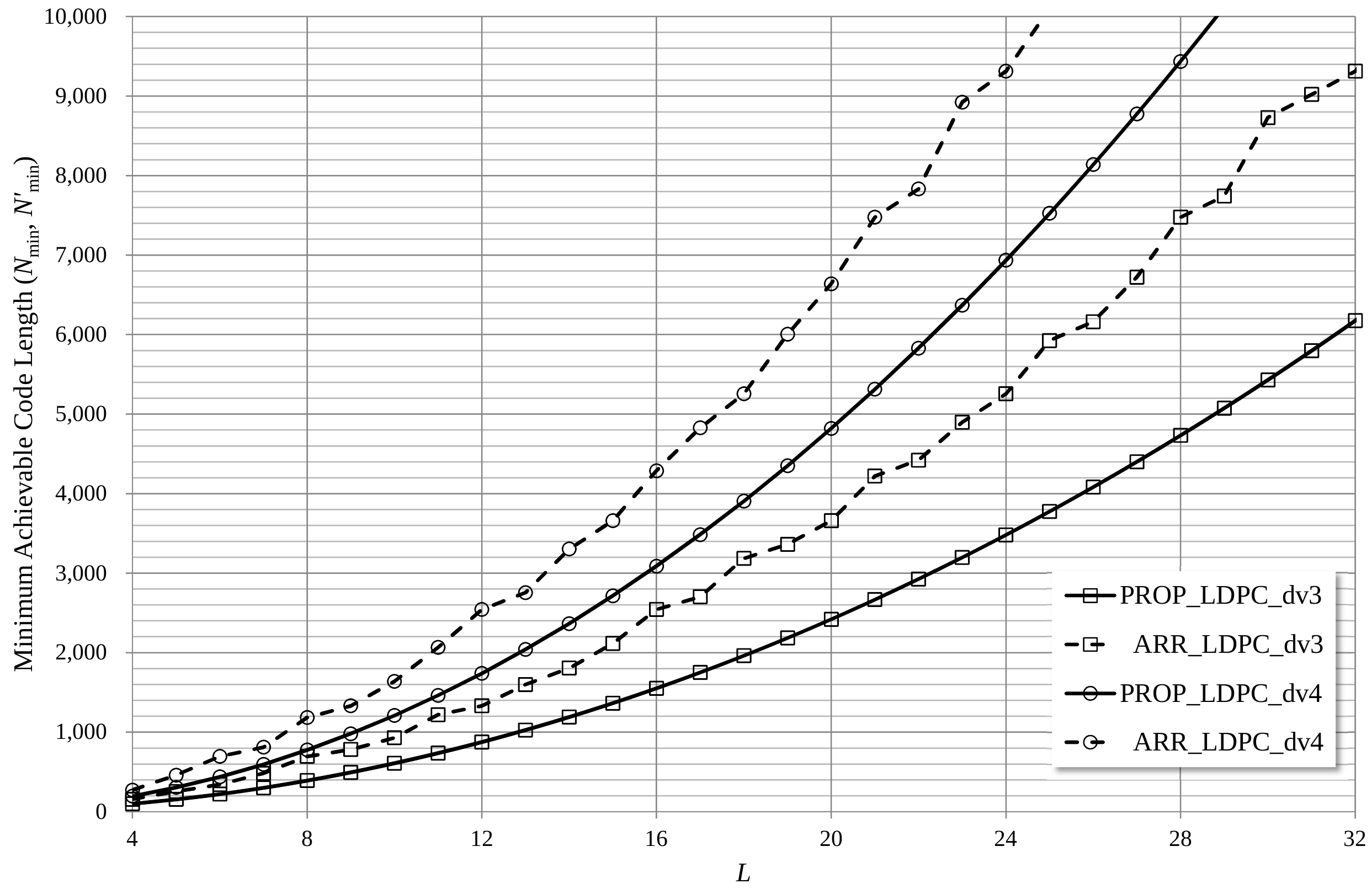}
	\caption{The minimum achievable code length of the proposed QC LDPC codes (denoted by PROP\_LDPC) and the array LDPC codes (denoted by ARR\_LDPC) for the girth 6.}
	\label{fig:min_length}
\end{figure*}

It is obvious that the proposed QC LDPC codes can achieve the minimum code length, which corresponds to the theoretical lower bound, among all possible regular LDPC codes for given $d_v$ and design rate under the girth 6 is guaranteed because the parity-check matrix of the proposed QC LDPC code with $z=d_v (d_v-1)L+1$ is actually an incidence matrix of a Steiner system \cite{Vasic}, \cite{Johnson3}. That is, the minimum achievable code length of the proposed QC LDPC codes is $N_{\min}=d_v (d_v-1)L^2+L$.
For comparison, consider QC LDPC codes whose parity-check matrices have the form of a $d_v \times d_v L$ array of $z' \times z'$ circulant permutation matrices, which correspond to the most common form of the existing regular QC LDPC codes. Obviously, these QC LDPC codes have the same design rate, row- and column-weights of the parity-check matrices with the proposed QC LDPC codes. The necessary condition on $z'$ for these QC LDPC codes to have the girth larger than or equal to 6 is $z' \geq d_v L$ \cite{Fossorier} and the array LDPC codes \cite{Fan} are known to achieve the equality \cite{Karimi}. Thus, the minimum achievable code length of these QC LDPC codes for guaranteeing the girth 6 is $N_{\min}'=(d_v)^2 L^2$. Fig. \ref{fig:min_length} illustrates such minimum achievable code lengths $N_{\min}$ and $N_{\min}'$ by varying $L$ for $d_v=3$ and $4$ and we can see that the proposed QC LDPC codes can flexibly have a very short length up to the theoretical lower bound unlike the array LDPC codes.
Note that QC LDPC codes in \cite{Xia} and \cite{Baldi1} cannot achieve the minimum code length.

The error correcting performance of the proposed QC LDPC codes may not be good for $z$ much larger than $d_v(d_v-1)L+1$ because, regardless of the code length, the girth of the proposed QC LDPC codes is fixed to 6 and the minimum distance of these QC LDPC codes has a value between $d_v+1$ and $2d_v$ \cite{Vasic}, \cite{Xia}. However, these restrictions on the girth and the minimum distance are not problematic for the proposed high-rate short QC LDPC codes and for large $L$ and small $z$, the error correcting performance of the proposed QC LDPC codes will be compared with that of other QC LDPC codes in Section \ref{sec:simulation}. Therefore, the proposed construction is adequate for high-rate QC LDPC codes of short and moderate lengths.

It is difficult to analyze the rank of the parity-check matrices of the proposed QC LDPC codes because they do not have an algebraic structure like the codes in \cite{Kou}-\cite{Kamiya}. Instead, the rank of the proposed parity-check matrices for various parameters can be numerically computed. 
It is observed that the parity-check matrices of the proposed QC LDPC codes have full rank for the parameters $d_v=3$, $4 \leq L \leq 20$, and $z$ such that the code length is less than or equal to 3,000, and have almost full rank, i.e., just one redundant row, for the parameters $d_v=4$, $4 \leq L \leq 15$, and $z$ such that the code length is less than or equal to 3,000. Moreover, every parity-check matrix has at least one full-rank circulant for $d_v=3$, which enables a simple encoding of the proposed QC LDPC codes, and has at least one almost full-rank circulant for $d_v=4$.
Assume that $H_L$ in (\ref{eq:H}) is invertible. Then, a generator matrix $G$ of systematic form is simply obtained \cite{Johnson2} as
\begin{equation*}
G = \begin{bmatrix}
~~ & & ~ & (H_L^{-1} H_1)^T \\
& I_{z(L-1)} & & (H_L^{-1} H_2)^T \\
& & & \vdots \\
& & & (H_L^{-1} H_{L-1})^T 
\end{bmatrix}
\end{equation*}
where $I_{z(L-1)}$ represents the $z(L-1) \times z(L-1)$ identity matrix.
This full-rank property of the parity-check matrices differentiates the proposed QC LDPC codes from the QC LDPC codes in \cite{Tang} and \cite{Kamiya}, whose parity-check matrices consist of a single row of circulants and have many redundant rows.

\begin{figure*}[tb]
	\centering
	\includegraphics[scale=0.6]{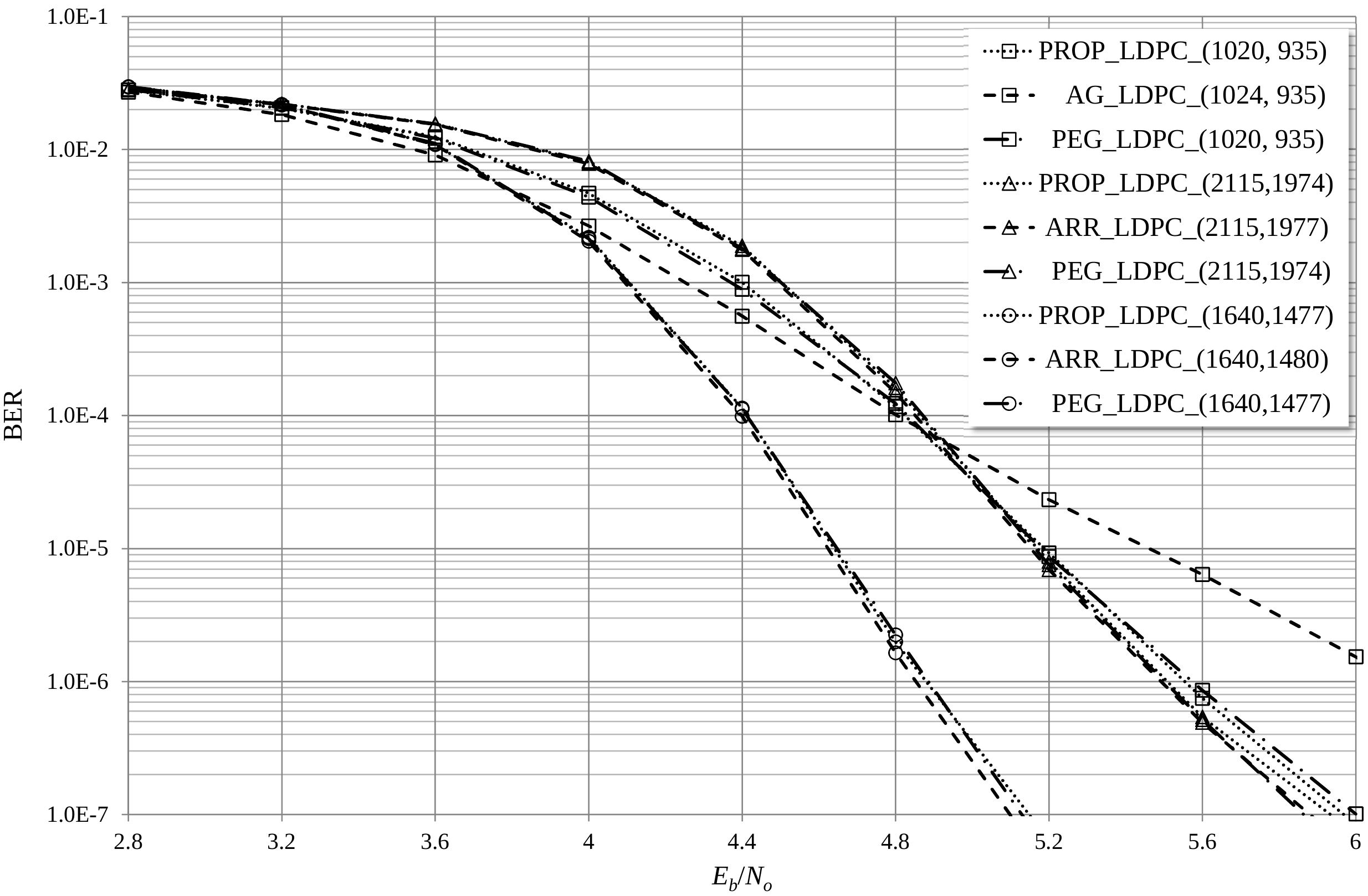}
	\caption{BER performance comparison of the proposed QC LDPC codes (denoted by PROP\_LDPC), the affine geometry QC LDPC codes (denoted by AG\_LDPC), the array LDPC codes (denoted by ARR\_LDPC), and the PEG LDPC codes (denoted by PEG\_LDPC).}
	\label{fig:performance}
\end{figure*}

\vspace{2mm}
\section{Simulation Results} \label{sec:simulation}

In this section, the error correcting performance of the proposed QC LDPC codes is verified via numerical analysis and compared with that of some algebraic QC LDPC codes and progressive edge-growth (PEG) LDPC codes \cite{Hu} with girth 6. As algebraic QC LDPC codes, affine geometry QC LDPC codes \cite{Kamiya} and array LDPC codes \cite{Fan} are used. Note that the PEG LDPC codes are not quasi-cyclic but random-like, and they are known to have the error correcting performance as good as random LDPC codes.
The parameters of the algebraic QC LDPC codes are set to have as equal values with those of the proposed QC LDPC codes as possible and the parameters of the PEG LDPC codes are exactly the same as those of the proposed QC LDPC codes. All results are obtained based on PC simulation using the sum-product decoding under the additive white Gaussian noise (AWGN) channel. The maximum number of iterations is set to 100.

First, the rate-0.9167 (1020, 935) proposed QC LDPC code with $d_v=3$, $L=12$, and $z=85$ is compared with the rate-0.9131 (1024, 935) affine geometry QC LDPC code and the rate-0.9167 (1020, 935) PEG LDPC code. The bit error rate (BER) performance of these LDPC codes is shown in Fig. \ref{fig:performance} and we can see that the proposed QC LDPC code and the PEG LDPC code show a better BER performance than the affine geometry QC LDPC code in the high signal-to-noise ratio (SNR) region. Second, the rate-0.9333 (2115, 1974) proposed QC LDPC code with $d_v=3$, $L=15$, and $z=141$ is compared with the rate-0.9348 (2115, 1977) array LDPC code and the rate-0.9333 (2115, 1974) PEG LDPC code. It is shown in Fig. \ref{fig:performance} that these LDPC codes have almost the same BER performance. Finally, the rate-0.9006 (1640, 1477) proposed QC LDPC code with $d_v=4$, $L=10$, and $z=164$ is compared with the rate-0.9024 (1640, 1480) array LDPC code and the rate-0.9006 (1640, 1477) PEG LDPC code. It is shown in Fig. \ref{fig:performance} that these LDPC codes also have almost the same BER performance.

%




%

%

\vspace{2mm}
\section{Conclusions} \label{sec:conclusions}

In this paper, a new class of high-rate QC LDPC codes with $d_v=3$ or $4$ is proposed, which have parity-check matrices consisting of a single row of circulants and having the girth 6. The construction of these QC LDPC codes exploits the CDFs constructed from hooked Skolem sequences in the case of $d_v=3$ and $L \equiv 2~\mathrm{or}~3 \mod 4$, and the PDFs in other cases. In designing the proposed QC LDPC codes, we can flexibly choose the values of $L$ and $z$ including the minimum achievable code length for a given design rate. The parity-check matrices of the proposed QC LDPC codes have full rank when $d_v=3$ and have almost full rank, i.e., just one redundant row, when $d_v=4$. Via numerical analysis, it is verified that the error correcting performance of the proposed QC LDPC codes is better than or almost equal to that of the affine geometry QC LDPC codes, the array LDPC codes, and the PEG LDPC codes.

\end{document}